\documentclass[12pt, anon]{l4dc2024}
\LinesNumbered
\DontPrintSemicolon

\title[Event-triggered Safe Bayesian Optimization on Quadcopters]{Event-Triggered Safe Bayesian Optimization on Quadcopters}
\usepackage{times}

\usepackage{bm}
\usepackage{tikz}
\usetikzlibrary{external}
\usepackage{booktabs}
\usepackage{mathtools}
\usepackage{algorithm}
\usepackage{paralist}
\usepackage{multicol}
\usepackage{wrapfig}

\usepackage{pgf}
\usepackage{pgfplots}
\pgfplotsset{compat=1.11}
\usepgfplotslibrary{units}
\DeclareUnicodeCharacter{2212}{\ensuremath{-}}
\usepgfplotslibrary{groupplots,dateplot}
\usetikzlibrary{patterns, shapes.arrows}
\pgfplotsset{compat=newest}

\definecolor{RWTHRed}{rgb}{0.8,0.027450980392156862,0.11764705882352941}

\newcommand{\R}{\mathbb{R}}
\newcommand{\N}{\mathbb{N}}

\newcommand{\ie}{i\/.\/e\/.,\/~}
\newcommand{\eg}{e\/.\/g\/.,\/~}
\newcommand{\cf}{cf\/.\/~}
\newcommand{\iid}{i\/.i\/.d\/.\/~}
\newcommand{\wrt}{w\/.r\/.t\/.\/~}

\newcommand*{\ET}{ET\nobreakdash-GP\nobreakdash-UCB}
\newcommand{\fakepar}[1]{\vspace{1mm}\noindent\textbf{#1.}}

\newcommand{\argmax}{\arg\max}
\newcommand{\optvar}{\bm{\theta}}
\newcommand{\obs}{\hat{J}}
\newcommand{\triggerT}{t'}
\newcommand{\timepar}{\xi}
\newcommand{\tvar}{\tau}

\usepackage{pifont}

\SetCommentSty{mycommfont}

\newtheorem{assumption}{Assumption}

\begin{document}

\author{%
 \Name{Antonia Holzapfel\nametag{\thanks{The authors contributed equally to this work.}
  \addtocounter{footnote}{-1}\addtocounter{Hfootnote}{-1}}} \Email{antonia.holzapfel@dsme.rwth-aachen.de}\\
 \Name{Paul Brunzema\nametag{\footnotemark}} \Email{paul.brunzema@dsme.rwth-aachen.de}\\
 \Name{Sebastian Trimpe} \Email{trimpe@dsme.rwth-aachen.de}\\
 \addr Institute for Data Science in Mechanical Engineering, RWTH Aachen University, Germany%
}

\maketitle

\begin{abstract}%
Bayesian optimization (BO) has proven to be a powerful tool for automatically tuning control parameters without requiring knowledge of the underlying system dynamics.
Safe BO methods, in addition, guarantee safety during the optimization process, assuming that the underlying objective function does not change.
However, in real-world scenarios, time-variations frequently occur, for example, due to wear in the system or changes in operation.
Utilizing standard safe BO strategies that do not address time-variations can result in failure as previous safe decisions may become unsafe over time, which we demonstrate herein.
To address this, we introduce a new algorithm, Event-Triggered SafeOpt (ETSO), which adapts to changes online solely relying on the observed costs.
At its core, ETSO uses an event trigger to detect significant deviations between observations and the current surrogate of the objective function. When such change is detected, the algorithm reverts to a safe backup controller, and exploration is restarted. 
In this way, safety is recovered and maintained across changes.
We evaluate ETSO on quadcopter controller tuning, both in simulation and hardware experiments.
ETSO outperforms state-of-the-art safe BO, achieving superior control performance over time while maintaining safety.
\end{abstract}

\section{Introduction}
 
Bayesian optimization (BO) has emerged as a powerful tool to tune controllers without requiring an exact model of the physical system, \eg in robotics \citep{calandra2016bayesian,marco2016automatic} or automotive applications \citep{Neumann2018Autotuning}.
This is achieved by optimizing a black-box objective function that quantifies the closed-loop performance as a function of controller parameters.
This (unknown) function is often modeled as a Gaussian process (GP), which is then used to formulate an acquisition function guiding the optimization process.
Specifically, the acquisition function determines which control parameters to choose next trading off exploration and exploitation.

Unbounded exploration can lead to control parameters that are unsafe and can harm the hardware of the physical system or people operating it.
Therefore, \textit{safe} BO algorithms have been proposed \citep{sui2015safe,berkenkamp2016safe} and their efficacy has been demonstrated on various control systems, including quadcopters \citep{berkenkamp2021bayesian}.
Starting from an initial safe controller, these algorithms build up a set of control parameters that are assumed to be safe.
Safe BO considers no variations in the controlled system and thus, no variation in the performance function between experiments.  
However, the closed-loop performance may change over time, \eg due to changes in the system dynamics or in the reference.
As a consequence, the underlying model can become invalid, and once found optimal controllers may become sub-optimal.
In addition, such changes may shift the safe set, and parameters that used to be safe may become unsafe.
Hence, adaptation to these changes is vital
to ensure both: good performance and safety over time.

Time-varying BO has been proposed to address change over time \citep{bogunovic2016time,brunzema2022event}. However, existing methods focus on performance, yet do not consider safe exploration.
In general, aiming for a-priori safety guarantees under \emph{arbitrary} changes is an ill-posed problem, and some assumptions are needed.
In this work, we consider systems that remain in a fixed mode for a minimum of time.
Such lower bounds on the time between changes are common in switched systems literature, where they are referred to as \emph{dwell time} \citep{hespanha1999Stability,baumann2019fast}. Scenarios that display this behavior are, \eg those that involve abrupt, but infrequent changes.
In this paper, we propose to combine existing mechanisms for safe exploration with the recently proposed event-triggered BO approach \citep{brunzema2022event} to deal with changes.
Event-triggered BO detects change by comparing the expected performance with obtained measurements, without requiring an explicit model of the time variations.
We thus propose Event-Triggered SafeOpt (ETSO) as a new type of safe BO algorithm that can accommodate for time variations. 
Our algorithm learns for a fixed number of rounds until it obtains a suitable controller; a practical approach for various applications.
During the optimization and afterwards, the event trigger monitors the learned function for changes.
Once such changes are detected, ETSO reverts to a safe backup controller, resets the safe set, and restarts the exploration. 
Unlike previous contextual BO approaches \citep{König21adaptive}, ETSO does not require any contextual measurements that quantify the temporal change, but only the noisy evaluations of the performance function.
\begin{figure}[t!]
    \centering
    \includegraphics[width=\textwidth]{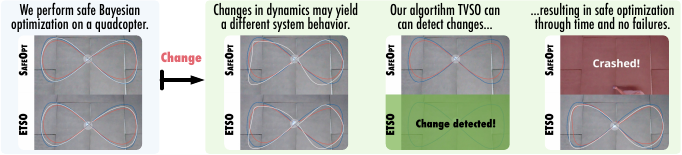}%
    \caption{We introduce ETSO, an algorithm for safe black-box optimization in time-varying environments. ETSO maintains safety across time variations in contrast to state-of-the-art algorithms such as SafeOpt.
    Images are from the video: \href{https://youtu.be/nLmeO-fMIvg}{https://youtu.be/nLmeO-fMIvg}.}
    \label{fig:hw_video}
\end{figure}

We test ETSO for controller tuning in simulation and on hardware.
Using a Bitcraze Crazyflie~2.1 quadcopter and experiments emulating typical variations, we show (\cf Figure~\ref{fig:hw_video}): \emph{(i)} standard safe BO can indeed fail due to time-variations, regardless of whether it is searching for or has already selected a controller, and \emph{(ii)}  ETSO remains safe across time variations in these situations.  
In summary, our main contributions are:
\begin{compactitem}
    \item \textbf{ETSO:} Combining safe BO and event triggering, we introduce the new algorithm ETSO that maintains safety across time variations by reverting to a backup controller and resetting its safe set when needed.
    \item \textbf{Empirical evaluation in simulation and hardware:} 
    ETSO shows superior safety and performance compared to baselines like SafeOpt, which can fail on time variations.
\end{compactitem}
\section{Problem Setting}\label{sec:problem}

We aim to find optimal control parameters $\optvar_t^* \in \Theta \subseteq \R^{d}$ over a time horizon $T$ in time-varying environments.\footnote{Throughout this paper, we highlight time dependency using a subscript $t \in \N$ to denote discrete time steps.}
The performance of a controller is quantified by the performance function ${J_{\timepar_t} : \Theta \to \R}$ where the functional dependency of $J_{\timepar_t}$ on $\optvar_t$ is unknown.
Here, we denote $\timepar_t$ as the system mode at time $t$ indicating that some variables of the closed loop system that influence the performance are time-varying.
We consider the setting where it is unknown \textit{which} variables can change, but the system has a minimum dwell time $\tau_{dwell}$ and the changes occur after some time steps since the last change $t'=t-\tau>\tau_{dwell}$, where $\tau$ is the time step of the change.
To avoid system failure, we further require the optimization procedure to be safe.
We therefore aim to ensure a minimum safety threshold $J_{\timepar_t}(\optvar_t) \geq J_{\mathrm{min},t}$ in probability where $J_{\mathrm{min},t} \in \R$ for all time steps $t \in \set{T}_{T} \coloneqq \{1, 2, \dots, T\}$ as it is standard in SafeOpt-type algorithms \citep{sui2015safe,berkenkamp2016safe}.
Note that in our setting this bound is also time-varying as it may depend on $\timepar_t$.
For safety, we assume access to a backup controller.
Such backup controllers are often supplied by the manufacturer
or obtained by some robust controller design: they are usually robust, but sub-optimal.
\begin{assumption} \label{ass:safe_backup}
    We have access to a safe backup controller $\optvar_B$ that stabilizes the close loop system for all system modes $\timepar_t$. A neighborhood of $\optvar_B$ is also safe such that $J_{\timepar_t}(\optvar_{\mathrm{B}})-\epsilon \geq J_{\mathrm{crit}}$.
\end{assumption}
In Assumption~\ref{ass:safe_backup}, $J_{\mathrm{crit}}\leq J_{\mathrm{min},t}$ is the critical performance before system failure and $\epsilon \in \R_{>0}$.
Remaining safe implies that we may not be able to find the global optimum, but only the optimum within the reachable set $\mathcal{R}_t \coloneqq \{\optvar \in \Theta \mid J_{\timepar_t}(\optvar) \geq J_{\mathrm{min},t} \text{ and $\optvar$ path-connected with $\optvar_B$}\}$.
Here we use $\optvar_{\mathrm{B}}$ from Assumption~\ref{ass:safe_backup} as an initial controller.
The reachable set may change depending on $\xi_t$; hence $\mathcal{R}_t$ is time-varying (see Figure~\ref{fig:toy_example}).
With $\mathcal{R}_t$, we can formalize our optimization problem as
\begin{equation}
    \optvar_t^* = \argmax_{\optvar \in \mathcal{R}_t} J_{\timepar_t}(\optvar).
    \label{eq:objective}
\end{equation}

At each time step $t$, an algorithm can query this performance function only once.
In the controller tuning context, this means performing one experiment \eg flying one round with a quadcopter (see Figure~\ref{fig:hw_video}).
The algorithm then receives an observation which is perturbed by independent and identically distributed (i.i.d.) zero mean Gaussian noise with $\sigma_n^2$ as the noise variance as
\begin{equation}
    \obs_t = J_{\timepar_t}(\optvar_t) + w_t, \quad w_t \sim \mathcal{N}(0, \sigma_n^2). \label{eq:noisy_obs}
\end{equation}

\fakepar{Problem Statement} We aim to develop a practical algorithm to tune controllers in time-varying environments by optimizing \eqref{eq:objective}.
The optimization should be performed without safety violations.
Contrary to related work in safe BO, we consider time variations in the underlying function, thus a changing reachable set $\mathcal{R}_t$, and no explicit measurements of the changing parameters.
\section{Related Work}\label{sec:related_work}

Our proposed algorithm ETSO builds on ideas from time-varying BO as well as safe BO to ensure safe learning over time. 
In the following, we discuss the related work to ETSO in more detail.

\fakepar{Safe Bayesian Optimization}
Safe BO methods aim to learn optimal parameters while also ensuring safety.
Most proposed algorithms follow the ideas of SafeOpt \citep{sui2015safe}.
Some popular examples of SafeOpt variants include modified SafeOpt \citep{berkenkamp2016safe}, multiple-constraint SafeOpt \citep{berkenkamp2021bayesian}, GoSafe \citep{baumann2021gosafe,sukhija2023gosafeopt}, and StageOpt \citep{pmlr-v80-sui18a}.
These variants have been used in various control applications, such as quadcopters \citep{berkenkamp2021bayesian} or heat pumps \citep{khoravi29sO}.
An alternative method, GoOSE \citep{turchetta2019safe}, can transform any algorithms to be safe by incorporating similar safety measures as SafeOpt.
However, all of these methods assume that the performance function is time-invariant.
This means that safe decisions are assumed to remain safe over time.
This does not hold in dynamic environments, where previously safe decisions may become unsafe.
We build on the modified SafeOpt method and extend it to time-varying environments.
It is worth noting that other approaches could be used; we predominantly choose modified SafeOpt due to computational efficiency, which is necessary due to the limited battery life of the quadcopters.

\fakepar{Adapting to Time-Variations}
There are several approaches to time-varying BO (TVBO), which consider trade-offs between remembering and forgetting data. These approaches either use a Markov-chain model~\citep{bogunovic2016time,brunzema2022controller,brunzema2022event} or a variation budget~\citep{zhou2021no, deng2022weighted} to model time variations.
Since we aim to be adaptive to changes, we build on the algorithm proposed in \cite{brunzema2022event} (\cf Section~\ref{sec:background}).
It leverages ideas from event-triggered learning  \citep{solowjow2020event,umlauft2019feedback}, \ie detecting changes online and re-learning only when necessary, and transfers them to TVBO.
Unlike all proposed TVBO approaches, we additionally consider safety and obtain an algorithm that learns safely under time variations such that it can be applied to hardware systems.

\fakepar{Contextual Safe Bayesian Optimization}
In contextual BO \citep{krause2011contextual}, the model of the underlying function includes a context parameter.
This parameter represents environmental conditions that cannot be influenced during optimization and may vary over time.
Safe contextual BO methods leverage this context parameter to account for time variations in the system and enable safe exploration using BO.
Such methods have been successfully applied to various control problems \citep{su18autonomous,fiducioso2019safe,de2023safe}.
Other methods, such as GoOSE for adaptive control \citep{König21adaptive} and VACBO \citep{xu2023violation}, are variants of this approach.
Given our problem setting in Section~\ref{sec:problem}, safe contextual BO methods cannot be applied as we assume to only have access to noisy observations and no additional environmental condition measurements.
To circumvent this lack of additional information to quantify change, our approach utilizes an event-based trigger to detect changes based on the expected and observed performance.
\section{Background}\label{sec:background}

Our method combines safe Bayesian optimization with the adaptation to time variations using an event trigger.
Next, we will introduce the necessary background and notation for both concepts.

\fakepar{Gaussian Processes}
GPs \citep{williams2006gaussian} are a probabilistic non-parametric method for regression that provide explicit uncertainty estimates for the learned function.
This makes them a powerful tool for regression and also as the probabilistic surrogate model for Bayesian optimization. 
A GP is fully defined by its mean function $m : \Theta \to \R$ and kernel $k : \Theta \times \Theta \to \R$ and we denote it as $f(\optvar) \sim \mathcal{GP}(m, k)$.
The prior belief defined through the mean function and kernel can be updated using a data set $\mathcal{D}_t\coloneqq\{(\optvar_i,\obs_i)\}_{i=1}^{t-1}$ to obtain a posterior prediction over a test point $\optvar$.
Assuming \iid additive Gaussian noise in the observations as in \eqref{eq:noisy_obs}, the posterior mean and covariance are $\mu_{\mathcal{D}_t}(\optvar)=m(\optvar)+\bm{k}_t^T(\optvar) (\mathbf{K}_t+\sigma_n^2 \mathrm{\mathbf{I}}_t)^{-1} (\bm{\hat{J}}-m(\optvar))$ and $\sigma^2_{\mathcal{D}_t}(\optvar)=k(\optvar,\optvar)-\bm{k}_t^T(\optvar) (\mathrm{\mathbf{K}}_t+\sigma_n^2 \mathrm{\mathbf{I}}_t)^{-1}\bm{k}_t(\optvar)$, respectively, where $\mathbf{K}_t = [k(\optvar_i, \optvar_j)]_{i, j = 1}^{t-1}$ is the Gram matrix, $\bm{k}_t(\optvar)= [k(\optvar_i, \optvar)]_{i=1}^{t-1}$, and noisy measurements are concatenated as $\bm{\hat{J}}=[\obs_1, \dots, \obs_{t-1}]^\top$.

\fakepar{Safe Bayesian Optimization}
BO is an optimization method for noisy black-box functions \citep{garnett2023bayesian}.
It searches for the global optimum as $\optvar^*=\argmax_{\optvar\in\Theta} J(\optvar)$.
The core components of BO are a surrogate model such as a GP and an acquisition function that balances exploration and exploitation.
To \textit{safely} optimize $J$, \ie without violating some minimal performance $J_{\mathrm{min}}$ with high probability, SafeOpt-type algorithms impose regularity assumptions on $J$:
first, that $J$ is $L$-Lipschitz continuous,
and second, that $J \in \mathcal{H}_k$, where $(\mathcal{H}_k, \|\cdot\|_k)$ is the unique reproducing kernel Hilbert space with the reproducing kernel $k$.
Specifically, SafeOpt-type algorithms require that there exists a bound $R \in \R_{\geq 0}$ such that $\|J\|_k \leq R$.
With this, one can leverage frequentist uncertainty bounds on the GP regression error which state that $|J(\optvar) - \mu_{\mathcal{D}_t}(\optvar)| \leq \beta_t \sigma_{\mathcal{D}_t}(\optvar)$ for all $\theta \in \Theta$ and all time steps holds true with probability at least $1 - \delta$ (see \eg \citet[Theorem~6]{srinivas2009gaussian}).
In our method, we use the modified SafeOpt algorithm \citep{berkenkamp2016safe} which directly computes the safe set based on the upper and lower confidence bounds based on such uncertainty bounds as $u_t(\optvar)\coloneqq\mu_{\mathcal{D}_t}(\optvar)+\beta_t\sigma_{\mathcal{D}_t}(\optvar)$ and $\ell_t(\optvar)\coloneqq\mu_{\mathcal{D}_t}(\optvar)-\beta_t\sigma_{\mathcal{D}_t}(\optvar)$, respectively.
With this, the safe set, the maximizers, and the expanders, of the modified SafeOpt algorithm are $ \mathcal{S}_t := \{\optvar \in \Theta \mid \ell_t(\optvar) \geq J_\mathrm{min}\}
$, $\mathcal{M}_t := \{ \optvar\in \mathcal{S}_t \mid u_t(\optvar) \geq \max_{\optvar' \in \mathcal{S}_t} \ell_t(\optvar')\} $ and $ \mathcal{G}_t := \{ \optvar\in \mathcal{S}_t \mid g_t(\optvar)>0\},\text{ with }g_t(\optvar):= \mid\{ \optvar'\in \Theta\setminus \mathcal{S}_t\mid \ell_{t, (\optvar,u_t(\optvar))}(\optvar')\geq J_\mathrm{min}\}\mid$, respectively. 
The acquisition function is then defined as $\optvar_t = \arg\max_{\optvar \in \mathcal{G}_t \cup \mathcal{M}_t} u_t(\optvar)-\ell_t(\optvar)$.
To start the optimization, SafeOpt-type algorithms require an initial safe controller as in Assumption~\ref{ass:safe_backup}.

\fakepar{Event-Triggered Time-Varying Bayesian optimization}
The concept of detecting and adapting to time-variations in BO using an event trigger was introduced by \cite{brunzema2022event}.
It allows their algorithm \ET{} to account for time variations without explicitly modelling and estimating a rate of change.
The event trigger is defined as follows.
\begin{definition}[Event-triggered TVBO framework]\label{def:event-trigger_framework}
    Given a test function $\psi_t$ and a threshold function $\kappa_t$, both of which can depend on the current dataset $\mathcal{D}_t$ and the latest query location and measurement pair $(\optvar_t, \obs_t)$, the event trigger at time step $t$ is defined as
    \begin{equation}
        \gamma_{\mathrm{reset}} = 1 \Leftrightarrow \psi_t\left(\mathcal{D}_t, (\optvar_t, \obs_t)\right) > \kappa_t\left(\mathcal{D}_t, (\optvar_t, \obs_t)\right)   \label{eq:event-trigger_framework}
    \end{equation}
    where $\gamma_{\mathrm{reset}}$ is the binary indicator for whether to reset the dataset $(\gamma_{\mathrm{reset}}=1)$ or not $(\gamma_{\mathrm{reset}}=0)$.
\end{definition}
\cite{brunzema2022event} define their test and threshold function such that the event trigger activates when observations deviate significantly from the prediction leveraging the following bound on deviations between measurements and the prediction where $t'$ is the time step since the last reset.
\begin{lemma}[{\citet[Lemma~3]{brunzema2022event}}]\label{lem:trigger_bound}
Let $J \sim \mathcal{GP}(0, k) $. Pick $\delta_B \in (0, 1)$ and set $\rho_{\triggerT} = 2\ln{\frac{2\pi_{\triggerT}}{\delta_B}}$, where $\sum_{{\triggerT} \geq 1} \pi_{\triggerT}^{-1} = 1$, $\pi_{\triggerT} >0$. Also set $\bar{w}_{\triggerT}^2 = {2\sigma_n^2 \ln \frac{2\pi_{\triggerT}}{\delta_B}}$. Then, observations $\obs_t $ satisfy 
$
    | \obs_t - \mu_{\mathcal{D}_t}(\optvar_t) | \leq \rho_{\triggerT}^{1/2} \sigma_{\mathcal{D}_t} (\optvar_t) + \bar{w}_{\triggerT}
$ for all time steps $\triggerT \geq 1$ with prob. at least $1-\delta$.
\end{lemma}
With this Lemma for non-changing functions $J$, they set the test function and threshold function to
\begin{align}
    \psi_t(\optvar_t) &\coloneqq \psi_t\left(\mathcal{D}_t, (\optvar_t, \obs_t)\right) = \left\lvert \obs_t - \mu_{\mathcal{D}_t}(\optvar_t) \right\rvert \label{eq:lower}\\
    \kappa_t(\optvar_t) &\coloneqq \kappa_t\left(\mathcal{D}_t, (\optvar_t, \obs_t)\right) = \sqrt{\rho_{\triggerT}} \sigma_{\mathcal{D}_t} (\optvar_t) + \bar{w}_{\triggerT}. \label{eq:upper_old}
\end{align}

\section{ETSO -- Combining Safety and Adaptation to Changing Environments}\label{sec:method}

In this section, we present our algorithm Event-Triggered SafeOpt (ETSO) for safe BO in time-varying environments.
\begin{figure}[t]
    \centering
    \includegraphics[width=\textwidth]{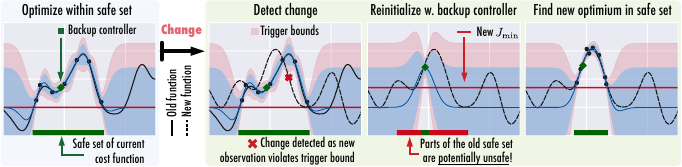}%
    \caption{Overview of our proposed algorithm. ETSO starts by optimizing a the performance function with SafeOpt. Once our event trigger detects a significant change, we reset to the safe backup controller, calculate a new threshold $J_{\text{min},t}$, and restart the exploration.}
    \label{fig:toy_example}
\end{figure}
Specifically, we aim to solve \eqref{eq:objective} for a given task by combining the concepts of SafeOpt-type algorithms and adapting to changes by detecting them based on the expected performance with an event trigger that can reset the data and thus, the safe set.
Figure~\ref{fig:toy_example} gives an overview of ETSO and it is summarized in Algorithm~\ref{alg:net}.
Initially, ETSO performs safe BO for safe exploration for $T_L \leq T$ rounds while being monitored by the event trigger.
This way, changes can already be detected during optimization.
As discussed, we choose the modified SafeOpt algorithm \citep{berkenkamp2016safe} for this step, as well as an event trigger as in Definition~\ref{def:event-trigger_framework} with \eqref{eq:lower} and \eqref{eq:upper_old}.
After $T_L$ rounds, ETSO chooses the parameters that maximize the posterior mean within the safe set.
If at any point the trigger is activated, ETSO resets the dataset, and explores the new system mode $\timepar_t$ with SafeOpt.
For this, the first round is performed using $\optvar_{\mathrm{B}}$ and a new limit to the cost $J_{\mathrm{min},t}$ with respect to it, as will be described in the implementation details section.

To prove safety for ETSO for different modes $\timepar_t$, we require assumptions on the performance of our event trigger as well as on the performance of the controller for the iteration in the new modes. 
\begin{assumption}
    \label{ass:imm_fail}
    When there is a change relevant for the performance, it is detected. Let the change happen at time $\tau$. The parameters $\optvar_{\tvar}$ are at least critically stable in the new system mode $\timepar_{\tvar}$.
\end{assumption}
The first part of the assumption is satisfied if the change is significant enough \wrt the threshold function of our event trigger.
This threshold can be large during optimization given our acquisition function that selects control parameters with the highest variance in the safe set.
It is substantially smaller after $T_L$ rounds, when ETSO selects its best predicted control parameters.
We will demonstrate in our experiments that not all changes can nor should be detected, as some are insignificant.
The second part of the assumption is a natural extension of Assumption~\ref{ass:safe_backup} and acts as a regularity assumption on the temporal change.
With the two assumptions, the following is immediate:

\begin{corollary}
    Given Assumption~\ref{ass:safe_backup} and Assumption~\ref{ass:imm_fail}, ETSO remains safe across time variations.
\end{corollary}
\begin{proof}
    We optimize safely with high probability using SafeOpt as shown in \citet[Theorem~1]{sui2015safe}. With Assumption~\ref{ass:imm_fail}, once a change occurs,  we detect it and have no system failure. ETSO resets the safe set and reverts to $\optvar_{\mathrm{B}}$, which is safe for the new system mode by Assumption~\ref{ass:safe_backup}. The following optimization is again safe by the above arguments proving the claim.
\end{proof}
This Corollary is an immediate extension of the safety proof from \cite{sui2015safe} to the time-varying case with an event trigger under the stated assumptions. Notwithstanding, it formalizes the conditions where ETSO is guaranteed to be safe and underscores the significance of a reliable event trigger. Essentially, the proposed event trigger allows for decoupling the entire process into individual, stationary sections, for which the safety rationale of existing safe BO can be applied.
We will demonstrate in our experiments, that ETSO will detect significant changes.
We will next describe implementation details regarding the event trigger and the new initialization after a reset.

\begin{algorithm2e}[b!]
\caption{Event-Triggered SafeOpt (ETSO).}
\label{alg:net}
\KwIn{GP prior ($\mu_{\mathcal{D}_0}, k$), backup controller $\optvar_{\mathrm{B}}$, $\delta_B\in(0,1)$, $\beta$,  max. learn rounds $T_L > 1$}
Initialize GP with $\mathcal{D}_1=(\optvar_{\mathrm{B}}, \obs_{\mathrm{B}})$; Set $J_{\mathrm{min},t} \gets J_{\mathrm{min}}(\obs_{\mathrm{B}})$; Set $t' \gets 1$\;
    \begin{multicols}{2}
    \For{$t\leftarrow 1$ \KwTo $T$}{
        Update GP with $\mathcal{D}_t $\;
        \uIf{$t'\leq T_L$}{
            $\mathcal{S}_t \gets \{\optvar \in \mathcal{X} \mid l_t(\optvar) \geq J_{\mathrm{min},t}\}$\;
            $\mathcal{M}_t \gets \{ \optvar \in \mathcal{S}_t \mid u_t(\optvar) \geq \underset{\optvar' \in \mathcal{S}_t}{\max} l_t(\optvar')\} $\;
            $\mathcal{G}_t \gets \{ \optvar\in \mathcal{S}_t \mid g_t(\optvar)>0\}$\;
            
        }
        \uIf{$t'< T_L$}{
        $\optvar_t \gets \argmax_{\optvar \in \mathcal{G}_t\cup \mathcal{M}_t}(a_t(\optvar)) $\;
        }
        \uElse{
        $\optvar_t \gets \argmax_{\optvar \in \mathcal{S}_t}(\mu_t(\optvar)) $\;
        }
        $\obs_t \gets $\textsc{Observation}$(\optvar_t)$\;
        $\mathcal{D}_{t+1}, t' \gets $\textbf{Event Trigger}$(\mathcal{D}_t, (\obs_t,\optvar_t), \delta_B)$\;
        }

\columnbreak
    \SetKwProg{myproc}{Event Trigger}{$(\mathcal{D}_t, (\obs_t,\optvar_t), \delta_B)$\string:}{{\KwRet $\mathcal{D}_{t+1}, t'$}}
    \myproc{}{
    $\gamma_{reset} \gets$ evaluate \eqref{eq:event-trigger_framework} with \eqref{eq:lower} and \eqref{eq:scaled_upper}\;
        \uIf{$\gamma_{reset}$}{
            \tcc{reset to backup controller}
            $\obs_{\mathrm{B}} \gets $\textsc{Observation}$(\optvar_{\mathrm{B}})$\;
            $\mathcal{D}_{t+1}=\{(\optvar_{\mathrm{B}}, \obs_{\mathrm{B}}), (\optvar_t,\obs_t)\}$\;
            $J_{\mathrm{min},t} \gets J_\mathrm{min}(\obs_{\mathrm{B}})$\;
            $t' \gets 2 $\;
        }
        \uElse{
            $\mathcal{D}_{t+1}\gets \mathcal{D}_t \cup {(\optvar_t,\obs_t)}$\;
            $t' \gets t' + 1$\;
        }
    }
    \end{multicols}
\end{algorithm2e}

\fakepar{Implementation details}
We next describe key algorithmic choices in our implementation of ETSO\footnote{The repository is freely available in \href{https://github.com/antoHolz/ETSO}{https://github.com/antoHolz/ETSO}.}.
First, the data cannot be standardized for the GP, because each data point is added to the SafeOpt algorithm individually, and all sets are updated at each time step.
Instead, we normalize the data with the rounded up absolute value of $\obs_{\mathrm{B}}$ and set the prior mean $\mu_{\mathcal{D}_0} \coloneqq m(\optvar) = -1$, obtaining the relation $-1=\mu_{\mathcal{D}_0} \leq \mu_{\mathcal{D}_1}(\optvar_{\mathrm{B}}) \leq \obs_{\mathrm{B}} / \lceil \text{abs}(\obs_{\mathrm{B}}) \rceil$. 
From SafeOpt we also derive some restrictions on the safety threshold and the prior standard deviation $\sigma_{\mathcal{D}_0}$. 
$J_{\mathrm{min},t}$ must be below the performance of the backup controller
and above the lower bound $\ell_t(\optvar)$ such that points for which we have no relevant information cannot be classified as safe.
This gives us the relation
\begin{align}
\label{eq:limJ}
    \mu_{\mathcal{D}_0}-\beta_0 \cdot \sigma_{\mathcal{D}_0}  < J_{\mathrm{min},t} < \mu_{\mathcal{D}_1}(\optvar_{\mathrm{B}})-\beta_1\cdot \sigma_{\mathcal{D}_1}(\optvar_{\mathrm{B}}),
\end{align}%
and we choose $J_{\mathrm{min},t}$ and $\sigma_{\mathcal{D}_0}$ as $J_{\mathrm{min},t} = \mu_{\mathcal{D}_0}-\beta_1\cdot \sigma_n - \epsilon$ and $\sigma_{\mathcal{D}_0} = \max( \tfrac{\mu_{\mathcal{D}_0}-J_{\mathrm{min},t}+\epsilon}{\beta_0}, \frac{1}{3} )$, respectively, as well as $\epsilon=0.2$.
Additionally, we scale the threshold function of the trigger as
\begin{align}
    \kappa_t\left(\mathcal{D}_t, (\optvar_t, \obs_t)\right) = 3/4 \cdot \sqrt{\rho_{\triggerT}} \sigma_{\mathcal{D}_t} (\optvar_t) + 1/4 \cdot \bar{w}_{\triggerT}. \label{eq:scaled_upper}
\end{align}
The reason for this is that the infinite series in Lemma~\ref{lem:trigger_bound} will result in large bound over time--especially in the noise term $\bar{w}_{\triggerT}$.
While this is necessary for theoretical reasons \citep{srinivas2009gaussian,bogunovic2016time,brunzema2022event}, it has practical downsides regarding safety as some changes may not be detected in later rounds.
By having the tighter bound in \eqref{eq:scaled_upper} that still increases in time, we increase safety in practice by reverting to $\optvar_{\mathrm{B}}$ more quickly.

\section{Experiments on a Quadcopter}

\begin{wrapfigure}[12]{r}{0.38\textwidth}
    \centering
    \vspace{-1.4cm}
    \includegraphics[width=0.38\textwidth]{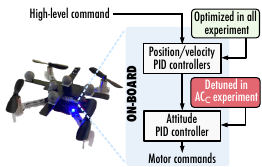}%
    \caption{Control structure on the Crazyflie 2.1. We optimize the PID gains of the position controller.}
    \label{fig:quadcopter_arch}
\end{wrapfigure}
We demonstrate our algorithm in simulation and on hardware of a quadcopter.
Specifically, we use a Bitcraze Crazyflie 2.1 with the Crazyswarm firmware~\citep{crazyswarm} and use a Vicon camera system to obtain position estimates.
For the simulation, we use the gym-pybullet-drones~\citep{panerati2021learning} environment.
The controller of the Crazyflie is a cascaded PID~\citep{mellinger2011minimum} as in Figure~\ref{fig:quadcopter_arch}.
We only tune the PI gains of the position PID, \ie $\optvar = [P_{xy}, P_z, I_{xy}, I_z] \in \R^4$.
As the backup controller, we set $\optvar_{\mathrm{B}}=[0.4,1.25,0.05,0.05]$ as specified in the firmware.
For the GP, we use a Matérn kernel ($\nu=5/2$), lenghtscales $[0.15, 0.75, 0.025, 0.050]$ as obtained from simulation, and set the noise standard deviation to $\sigma_n = 0.016$ based on measurements obtained from the real system.
Additionally, we use the parameter bounds $[0,10]$ for each parameter, thus including many unstable configurations in the search space. 
We use $\beta_\mathrm{const}=2$ as the $\beta$ hyperparameter for SafeOpt; a value that is widely used in the literature (\eg\cite{berkenkamp2016safe}).
Alternatively, a logarithmic $\beta$ can be used, which increases with time \citep{kandasamy2015high}.
This more conservative approach can improve safety.
We choose $\beta_\mathrm{const}$ because it leads to faster learning, which is necessary given the quadcopter's limited battery life.
We choose for the optimization the following cost functions for simulation and hardware, respectively,
\begin{align}\label{eq:cost}
    J_{\text{sim}}(\optvar)=-\sum_{q=1}^Q |e_{xyz}(q)|^\frac{1}{2},   \:\:  J_{\text{hw}}(\optvar)=\-\sum_{q=1}^Q |e_{xyz}(q)|^\frac{1}{4},
\end{align}
where $e_{xyz}(q)= \min_{k\in \mathcal{T}_{K}}{ \sqrt{(x_k-x_q)^2+(y_k-y_q)^2+(z_k-z_q)^2}}$, for a trajectory with waypoint index $q \in \mathcal{T}_{Q}$ and position measurements at time steps $k\in \mathcal{T}_{K}$ within an experiment. Hence, we measure performance as the minimum deviation achieved by the quadcopter to pre-defined waypoints. 
We perform three different experiments.
In each of them, a variation occurs at time $\tvar > T_L$:
\begin{compactitem}
    \item \textbf{2D Trajectory variation (2DTV)}: the quadcopter flies a figure eight in the $xy$-plane, and after $\tvar$, an hourglass-like figure is specified as the reference.
    \item \textbf{3D Trajectory variation (3DTV)}: the quadcopter flies a figure eight in the $xy$-plane, and after $\tvar$, a figure eight with altitude changes is specified as the reference.
    \item \textbf{Change in attitude controller (AC$_C$)}: the quadcopter flies a figure eight in the $xy$-plane with the default attitude controller, and at $\tvar$, we change the gains of the attitude controller of the quadcopter by multiplying them with the factor $C$ (\cf Figure~\ref{fig:quadcopter_arch}).\footnote{In the output of an attitude controller, the computed torques usually have the form (from~\cite{luukkonen2011modelling}):
        \begin{align}
            \tau_{\phi} &= \left(K_{\phi,P} \cdot e_{\phi}+ K_{\phi,I}\cdot \int e_{\phi}(t)+ K_{\phi,D}\cdot \tfrac{\partial e_{\phi}(t)}{\partial t}\right) \cdot I_{xx}
             \nonumber
        \end{align}
        where $I$ is the moments of inertia of the quadcopter. Here, changing all gains by some factor $C$ is equivalent to changing the moment of inertia $I_{xx}$. In other words, changing the low level controller has the same effect as having miss-specified moment of inertia in the controller. This is consistent with the visual impression (see video).}
\end{compactitem}
\begin{figure}[t!]
    \centering
    \includegraphics[width=\textwidth]{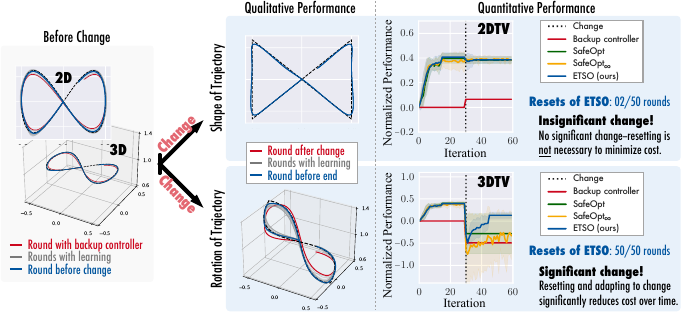}%
    \caption{Trajectory change experiments in simulation. Left: Trajectories of the quadcopter before the variation time. Center: Trajectories of the quadcopter after the variation time, with ETSO. Right: Performance of the algorithms normalized by the initial performance of the backup controller as $\left(J_{\xi_0}(\bm{\theta}_B) - J_{\xi_t}(\bm{\theta}_t)\right) / J_{\xi_0}(\bm{\theta}_B)$ and the resets performed by ETSO.}
    \label{fig:sim_traj}
\end{figure}
We execute the simulation experiments for $T_{\mathrm{sim}}=60$ rounds, with $T_{L,\mathrm{sim}}=15$ rounds of exploration for the finite learning algorithms.
Due to the limited battery life of the quadcopters, we perform fewer rounds in the hardware experiments.
For the 2DTV and AC experiments, we define $T_{\mathrm{hw}}=20$ and $T_{L,\mathrm{hw}}=9$, and for the 3DTV experiment, $T_{\mathrm{hw}}=16$ and $T_{L,\mathrm{hw}}=7$.
We compare our algorithm \textbf{ETSO} to \textbf{SafeOpt} using the specified budgets above.
Further, we compare to \textbf{SafeOpt$_\infty$} as a variant with unlimited budget as well as to the performance of the safe \textbf{backup controller}.

\subsection{Simulation Results}
Figure~\ref{fig:sim_traj} shows the results of the trajectory change experiments 2DTV and 3DTV.
Both experiments have the same initial trajectory (a figure eight in the 2D plane) and undergo a change at $\tvar=30$.
The right side of Figure~\ref{fig:sim_traj} shows the performance of the algorithms for each experiment setup.
For the experiment 2DTV (top), the change in reference results in an insignificant change in performance.
With this insignificant change, the event trigger does not activate; hence, ETSO does not reset, saving resources compared to re-optimizing.
Here, all baselines outperform the backup controller and yield similar performance.
In the experiment 3DTV (bottom), ETSO outperforms SafeOpt$_\infty$ and SafeOpt.
Changing the reference in altitude results in a significant change in the performance function.
ETSO detects this change, resets (see bottom right in Figure~\ref{fig:sim_traj}), and then re-optimizes without stale data.
We observe the same trend for ETSO and SafeOpt in Figure~\ref{fig:pic2} (left) for the AC$_{.65}$ simulation experiments. 
However, using SafeOpt$_\infty$, the quadcopter crashed 50 out of 50 runs after some time steps after the system change.
SafeOpt$_\infty$ assumes that areas that were previously safe remain safe, but a new system mode can shift the safe set.
By resetting to $\optvar_{\mathrm{B}}$, ETSO resets its safe set to ensure good performance as well as safe optimization over time with no crashes.
SafeOpt also results in no crashes.
Therefore, the system change in AC$_{.65}$ satisfied Assumption~\ref{ass:imm_fail}.

\subsection{Hardware Results}
The results of the hardware experiments displayed in Figure~\ref{fig:pic2} confirm the insights from the simulations.
Insignificant changes (2DTV) do not require re-learning.
Here, ETSO and SafeOpt perform equivalently as ETSO does not reset.
SafeOpt$_\infty$ exhibits marginally improved performance towards the end of the 2DTV experiment.
\begin{figure}[t]
    \centering
    \includegraphics[width=\textwidth]{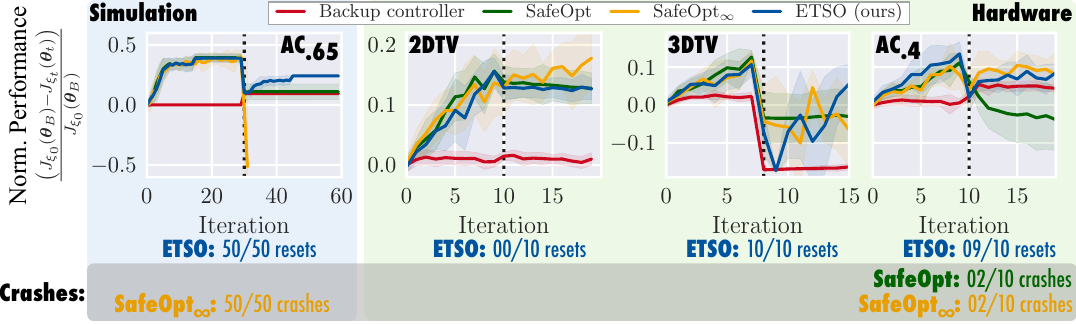}%
    \caption{Performance in the remaining simulation and hardware experiments (only runs without failures), normalized by the initial performance of the backup controller.
    }
    \label{fig:pic2}
\end{figure}
The training time for the other algorithms is insufficient to find the optimal control parameters.
This is unavoidable in the hardware setup due to battery constraints.
For the experiments with the altitude change (3DTV), we observe that ETSO outperforms SafeOpt, picking the best performing controller at the last time step.
In AC$_{.40}$, SafeOpt performs worse than in simulation.
Furthermore, we recorded two crashes out of the $10$ different rounds for SafeOpt.
This highlights that in practice iterations with the quadcopter are not fully \iid in the real world.
As a remedy, we included a buffer time in between the rounds, but some influence of the previous control parameters in the next iterations was inevitable.
ETSO and SafeOpt$_\infty$ yield similar performance, however, for SafeOpt$_\infty$ we recorded two out of 10 crashes.
Using ETSO, we recorded no crashes, highlighting that our event trigger increases safety in practice by resetting the safe set in presence of significant changes.
In Figure~\ref{fig:hw_video}, we showed a comparative example of the performance of ETSO and $\textrm{SafeOpt}_\infty$.
In the video, SafeOpt$_\infty$ fails two time steps after the change, while ETSO manages to reset to the backup controller and learn a better controller for the changed system.
\section{Concluding Remarks}
Event-Triggered SafeOpt (ETSO) is an algorithm that combines the concepts of SafeOpt with an event trigger as a first step toward safety in time-varying environments with only access to the cost signal. In our software and hardware experiments with a quadcopter, we find that:
\begin{compactenum}[\itshape(i)]
    \item ETSO only resets for significant changes. Without significant changes, ETSO performs equivalently to SafeOpt and SafeOpt$_\infty$ and better than the safe backup controller. 
    \item After detecting significant changes, ETSO outperforms the other baselines.
    \item If changes cause the safe set to shift considerably, using SafeOpt or SafeOpt$_\infty$ can lead to failures during optimization, whereas ETSO avoids such failures by resetting the safe set. 
\end{compactenum}
%
For some systems, it might not be possible to fulfill Assumption~\ref{ass:imm_fail}.
Here, using crash prevention strategies as proposed in GoSafe-variants~\citep{baumann2021gosafe,sukhija2023gosafeopt} can help.

\acks{We thank A. Gräfe and S. Giedyk for their help with the hardware setup, and H. Hose for insightful discussions on the experiment design.
We further thank A. von Rohr and J. Menn for helpful comments on the manuscript.
This work is partially funded by the Deutsche Forschungsgemeinschaft 
(DFG, German Research Foundation) – RTG 2236/2 (UnRAVeL) and DFG Priority Program SPP~2422 (TR 1433/3-1).} 

\bibliography{references}

\begin{thebibliography}{34}
\providecommand{\natexlab}[1]{#1}
\providecommand{\url}[1]{\texttt{#1}}
\expandafter\ifx\csname urlstyle\endcsname\relax
  \providecommand{\doi}[1]{doi: #1}\else
  \providecommand{\doi}{doi: \begingroup \urlstyle{rm}\Url}\fi

\bibitem[Baumann et~al.(2019)Baumann, Mager, Jacob, Thiele, Zimmerling, and Trimpe]{baumann2019fast}
Dominik Baumann, Fabian Mager, Romain Jacob, Lothar Thiele, Marco Zimmerling, and Sebastian Trimpe.
\newblock Fast feedback control over multi-hop wireless networks with mode changes and stability guarantees.
\newblock \emph{ACM Transactions on Cyber-Physical Systems}, 4\penalty0 (2):\penalty0 1--32, 2019.

\bibitem[Baumann et~al.(2021)Baumann, Marco, Turchetta, and Trimpe]{baumann2021gosafe}
Dominik Baumann, Alonso Marco, Matteo Turchetta, and Sebastian Trimpe.
\newblock Gosafe: Globally optimal safe robot learning.
\newblock In \emph{2021 IEEE International Conference on Robotics and Automation (ICRA)}, pages 4452--4458. IEEE, 2021.

\bibitem[Berkenkamp et~al.(2016)Berkenkamp, Schoellig, and Krause]{berkenkamp2016safe}
Felix Berkenkamp, Angela~P Schoellig, and Andreas Krause.
\newblock Safe controller optimization for quadrotors with {Gaussian} processes.
\newblock In \emph{2016 IEEE International Conference on Robotics and Automation (ICRA)}, pages 491--496. IEEE, 2016.

\bibitem[Berkenkamp et~al.(2021)Berkenkamp, Krause, and Schoellig]{berkenkamp2021bayesian}
Felix Berkenkamp, Andreas Krause, and Angela~P Schoellig.
\newblock {Bayesian} optimization with safety constraints: safe and automatic parameter tuning in robotics.
\newblock \emph{Machine Learning}, pages 1--35, 2021.

\bibitem[Bogunovic et~al.(2016)Bogunovic, Scarlett, and Cevher]{bogunovic2016time}
Ilija Bogunovic, Jonathan Scarlett, and Volkan Cevher.
\newblock Time-varying {Gaussian process} bandit optimization.
\newblock In \emph{Artificial Intelligence and Statistics}, pages 314--323. PMLR, 2016.

\bibitem[Brunzema et~al.(2022)Brunzema, Von~Rohr, and Trimpe]{brunzema2022controller}
Paul Brunzema, Alexander Von~Rohr, and Sebastian Trimpe.
\newblock On controller tuning with time-varying {Bayesian} optimization.
\newblock In \emph{2022 IEEE 61st Conference on Decision and Control (CDC)}, pages 4046--4052. IEEE, 2022.

\bibitem[Brunzema et~al.(2023)Brunzema, von Rohr, Solowjow, and Trimpe]{brunzema2022event}
Paul Brunzema, Alexander von Rohr, Friedrich Solowjow, and Sebastian Trimpe.
\newblock Event-triggered time-varying {Bayesian} optimization.
\newblock \emph{arXiv preprint arXiv:2208.10790}, 2023.

\bibitem[Calandra et~al.(2016)Calandra, Seyfarth, Peters, and Deisenroth]{calandra2016bayesian}
Roberto Calandra, Andr{\'e} Seyfarth, Jan Peters, and Marc~Peter Deisenroth.
\newblock {Bayesian} optimization for learning gaits under uncertainty: An experimental comparison on a dynamic bipedal walker.
\newblock \emph{Annals of Mathematics and Artificial Intelligence}, 76:\penalty0 5--23, 2016.

\bibitem[De~Blasi et~al.(2023)De~Blasi, Bahrami, Engels, and Gepperth]{de2023safe}
Stefano De~Blasi, Maryam Bahrami, Elmar Engels, and Alexander Gepperth.
\newblock Safe contextual {Bayesian} optimization integrated in industrial control for self-learning machines.
\newblock \emph{Journal of Intelligent Manufacturing}, pages 1--19, 2023.

\bibitem[Deng et~al.(2022)Deng, Zhou, Kim, Tewari, Gupta, and Shroff]{deng2022weighted}
Yuntian Deng, Xingyu Zhou, Baekjin Kim, Ambuj Tewari, Abhishek Gupta, and Ness Shroff.
\newblock Weighted {Gaussian} process bandits for non-stationary environments.
\newblock In \emph{International Conference on Artificial Intelligence and Statistics}, pages 6909--6932. PMLR, 2022.

\bibitem[Fiducioso et~al.(2019)Fiducioso, Curi, Schumacher, Gwerder, and Krause]{fiducioso2019safe}
Marcello Fiducioso, Sebastian Curi, Benedikt Schumacher, Markus Gwerder, and Andreas Krause.
\newblock Safe contextual bayesian optimization for sustainable room temperature {PID} control tuning.
\newblock In \emph{Proceedings of the Twenty-Eighth International Joint Conference on Artificial Intelligence}, pages 5850--5856. International Joint Conferences on Artificial Intelligence, 2019.

\bibitem[Garnett(2023)]{garnett2023bayesian}
Roman Garnett.
\newblock \emph{{Bayesian} Optimization}.
\newblock Cambridge University Press, 2023.

\bibitem[Hespanha and Morse(1999)]{hespanha1999Stability}
J.P. Hespanha and A.S. Morse.
\newblock Stability of switched systems with average dwell-time.
\newblock In \emph{Proceedings of the 38th IEEE Conference on Decision and Control}, pages 2655--2660 vol.3, 1999.

\bibitem[Kandasamy et~al.(2015)Kandasamy, Schneider, and P{\'o}czos]{kandasamy2015high}
Kirthevasan Kandasamy, Jeff Schneider, and Barnab{\'a}s P{\'o}czos.
\newblock High dimensional {Bayesian} optimisation and bandits via additive models.
\newblock In \emph{International Conference on Machine Learning}, pages 295--304. PMLR, 2015.

\bibitem[Khosravi et~al.(2019)Khosravi, Eichler, Schmid, Smith, and Heer]{khoravi29sO}
Mohammad Khosravi, Annika Eichler, Nicolas Schmid, Roy~S. Smith, and Philipp Heer.
\newblock Controller tuning by {Bayesian} optimization an application to a heat pump.
\newblock In \emph{2019 18th European Control Conference (ECC)}, pages 1467--1472, 2019.

\bibitem[Krause and Ong(2011)]{krause2011contextual}
Andreas Krause and Cheng Ong.
\newblock Contextual {Gaussian} process bandit optimization.
\newblock \emph{Advances in Neural Information Processing Systems}, 24, 2011.

\bibitem[König et~al.(2021)König, Turchetta, Lygeros, Rupenyan, and Krause]{König21adaptive}
Christopher König, Matteo Turchetta, John Lygeros, Alisa Rupenyan, and Andreas Krause.
\newblock Safe and efficient model-free adaptive control via {Bayesian} optimization.
\newblock In \emph{2021 IEEE International Conference on Robotics and Automation (ICRA)}, pages 9782--9788, 2021.

\bibitem[Luukkonen(2011)]{luukkonen2011modelling}
Teppo Luukkonen.
\newblock Modelling and control of quadcopter.
\newblock \emph{Independent research project in applied mathematics, Espoo}, 22\penalty0 (22), 2011.

\bibitem[Marco et~al.(2016)Marco, Hennig, Bohg, Schaal, and Trimpe]{marco2016automatic}
Alonso Marco, Philipp Hennig, Jeannette Bohg, Stefan Schaal, and Sebastian Trimpe.
\newblock Automatic {LQR} tuning based on {Gaussian} process global optimization.
\newblock In \emph{2016 IEEE International Conference on Robotics and Automation (ICRA)}, pages 270--277. IEEE, 2016.

\bibitem[Mellinger and Kumar(2011)]{mellinger2011minimum}
Daniel Mellinger and Vijay Kumar.
\newblock Minimum snap trajectory generation and control for quadrotors.
\newblock In \emph{2011 IEEE International Conference on Robotics and Automation}, pages 2520--2525. IEEE, 2011.

\bibitem[Neumann-Brosig et~al.(2018)Neumann-Brosig, Marco, Schwarzmann, and Trimpe]{Neumann2018Autotuning}
Matthias Neumann-Brosig, Alonso Marco, Dieter Schwarzmann, and Sebastian Trimpe.
\newblock Data-efficient autotuning with {Bayesian} optimization: An industrial control study.
\newblock \emph{IEEE Transactions on Control Systems and Technology}, 2018.

\bibitem[Panerati et~al.(2021)Panerati, Zheng, Zhou, Xu, Prorok, and Schoellig]{panerati2021learning}
Jacopo Panerati, Hehui Zheng, SiQi Zhou, James Xu, Amanda Prorok, and Angela~P. Schoellig.
\newblock Learning to fly---a gym environment with pybullet physics for reinforcement learning of multi-agent quadcopter control.
\newblock In \emph{2021 IEEE/RSJ International Conference on Intelligent Robots and Systems (IROS)}, 2021.

\bibitem[Preiss et~al.(2017)Preiss, H\"onig, Sukhatme, and Ayanian]{crazyswarm}
James~A. Preiss, Wolfgang H\"onig, Gaurav~S. Sukhatme, and Nora Ayanian.
\newblock Crazyswarm: {A} large nano-quadcopter swarm.
\newblock In \emph{{IEEE} International Conference on Robotics and Automation ({ICRA})}, pages 3299--3304. {IEEE}, 2017.

\bibitem[Solowjow and Trimpe(2020)]{solowjow2020event}
Friedrich Solowjow and Sebastian Trimpe.
\newblock Event-triggered learning.
\newblock \emph{Automatica}, 117:\penalty0 109009, 2020.

\bibitem[Srinivas et~al.(2010)Srinivas, Krause, Kakade, and Seeger]{srinivas2009gaussian}
Niranjan Srinivas, Andreas Krause, Sham Kakade, and Matthias Seeger.
\newblock {Gaussian} process optimization in the bandit setting: no regret and experimental design.
\newblock In \emph{Proceedings of the 27th International Conference on International Conference on Machine Learning}, pages 1015--1022, 2010.

\bibitem[Su et~al.(2018)Su, Wu, Cheng, and Chen]{su18autonomous}
Jie Su, Junfeng Wu, Peng Cheng, and Jiming Chen.
\newblock Autonomous vehicle control through the dynamics and controller learning.
\newblock \emph{IEEE Transactions on Vehicular Technology}, 67\penalty0 (7):\penalty0 5650--5657, 2018.

\bibitem[Sui et~al.(2015)Sui, Gotovos, Burdick, and Krause]{sui2015safe}
Yanan Sui, Alkis Gotovos, Joel Burdick, and Andreas Krause.
\newblock Safe exploration for optimization with {Gaussian} processes.
\newblock In \emph{International conference on machine learning}, pages 997--1005. PMLR, 2015.

\bibitem[Sui et~al.(2018)Sui, Zhuang, Burdick, and Yue]{pmlr-v80-sui18a}
Yanan Sui, Vincent Zhuang, Joel Burdick, and Yisong Yue.
\newblock Stagewise safe {B}ayesian optimization with {G}aussian processes.
\newblock In Jennifer Dy and Andreas Krause, editors, \emph{Proceedings of the 35th International Conference on Machine Learning}, volume~80 of \emph{Proceedings of Machine Learning Research}, pages 4781--4789. PMLR, 10--15 Jul 2018.

\bibitem[Sukhija et~al.(2023)Sukhija, Turchetta, Lindner, Krause, Trimpe, and Baumann]{sukhija2023gosafeopt}
Bhavya Sukhija, Matteo Turchetta, David Lindner, Andreas Krause, Sebastian Trimpe, and Dominik Baumann.
\newblock Gosafeopt: Scalable safe exploration for global optimization of dynamical systems.
\newblock \emph{Artificial Intelligence}, 320:\penalty0 103922, 2023.

\bibitem[Turchetta et~al.(2019)Turchetta, Berkenkamp, and Krause]{turchetta2019safe}
Matteo Turchetta, Felix Berkenkamp, and Andreas Krause.
\newblock Safe exploration for interactive machine learning.
\newblock \emph{Advances in Neural Information Processing Systems}, 32, 2019.

\bibitem[Umlauft and Hirche(2019)]{umlauft2019feedback}
Jonas Umlauft and Sandra Hirche.
\newblock Feedback linearization based on {Gaussian} processes with event-triggered online learning.
\newblock \emph{IEEE Transactions on Automatic Control}, 65\penalty0 (10):\penalty0 4154--4169, 2019.

\bibitem[Williams and Rasmussen(2006)]{williams2006gaussian}
Christopher~KI Williams and Carl~Edward Rasmussen.
\newblock \emph{{Gaussian} processes for machine learning}, volume~3.
\newblock MIT press Cambridge, MA, 2006.

\bibitem[Xu et~al.(2023)Xu, Jones, Svetozarevic, Laughman, and Chakrabarty]{xu2023violation}
Wenjie Xu, Colin~N Jones, Bratislav Svetozarevic, Christopher~R Laughman, and Ankush Chakrabarty.
\newblock Violation-aware contextual {Bayesian} optimization for controller performance optimization with unmodeled constraints.
\newblock \emph{arXiv preprint arXiv:2301.12099}, 2023.

\bibitem[Zhou and Shroff(2021)]{zhou2021no}
Xingyu Zhou and Ness Shroff.
\newblock No-regret algorithms for time-varying {Bayesian} optimization.
\newblock In \emph{2021 55th Annual Conference on Information Sciences and Systems (CISS)}, pages 1--6. IEEE, 2021.

\end{thebibliography}

\end{document}